\documentclass[11pt]{article}
\usepackage[margin=1.1in]{geometry}
\usepackage{booktabs} 
\usepackage{times}

\usepackage[backref,colorlinks,citecolor=blue,bookmarks=true]{hyperref}
\usepackage{mathtools, amssymb, amsthm, bbm, tabularx}
\usepackage{algorithmic}
\usepackage[ruled,vlined]{algorithm2e}
\usepackage{svg}
 \usepackage{tikz}
\usepackage{float}

\usepackage[title]{appendix}
\numberwithin{equation}{section}

\usepackage{mathtools, amssymb, amsthm, bbm}
\usepackage[nameinlink,capitalize]{cleveref}
\usepackage{enumitem}

\theoremstyle{plain}
\newtheorem{theorem}{Theorem}[section]
\newtheorem{lemma}[theorem]{Lemma}

\newtheorem*{claim}{Claim}

\theoremstyle{definition}
\newtheorem{definition}[theorem]{Definition}

\theoremstyle{remark}

\numberwithin{equation}{section}

\def\C{\mathcal{C}}
\def\D{\mathcal{D}}

\def\I{\mathcal{I}}

\def\L{\mathcal{L}}

\newcommand*{\N}{{\mathbb{N}}}

\newcommand*{\R}{{\mathbb{R}}}

\let\eps\epsilon
\let\phi\varphi

\DeclareMathOperator*{\pr}{\mathbb{P}}

\DeclareMathOperator*{\E}{\mathbb{E}}

\DeclareMathOperator{\unif}{Unif}

\DeclareMathOperator{\ind}{\mathbbm{1}}

\newcommand{\cube}[1]{\{\pm 1\}^{#1}}

\newcommand{\ignore}[1]{}

\newcommand*{\x}{\mathbf{x}}

\newcommand{\concept}{f}
\newcommand{\copt}{\concept^*}

\newcommand{\pup}{p_{\mathrm{up}}}
\newcommand{\pdown}{p_{\mathrm{down}}}

\newcommand{\nats}{\mathbb{N}}

\newcommand{\Unif}{\mathrm{Unif}}

\newcommand{\Sinp}{S_{\mathrm{inp}}}
\newcommand{\Sref}{S_{\mathrm{ref}}}
\newcommand{\Sadv}{S_{\mathrm{adv}}}
\newcommand{\Scln}{S_{\mathrm{cln}}}

\newcommand{\Sout}{S_{\mathrm{filt}}}

\newcommand{\Sinplabeled}{\bar{S}_{\mathrm{inp}}}

\newcommand{\Sadvlabeled}{\bar{S}_{\mathrm{adv}}}
\newcommand{\Sclnlabeled}{\bar{S}_{\mathrm{cln}}}
\newcommand{\Sremlabeled}{\bar{S}_{\mathrm{rem}}}
\newcommand{\Soutlabeled}{\bar{S}_{\mathrm{filt}}}

\newcommand{\coef}{\mathrm{coef}}

\title{Learning Constant-Depth Circuits in \\ Malicious Noise Models}
\author{
    Adam R. Klivans\thanks{\texttt{klivans@cs.utexas.edu}. Supported by NSF award AF-1909204 and the NSF AI Institute for Foundations of Machine Learning (IFML).} \\ UT Austin
	\and
    Konstantinos Stavropoulos\thanks{\texttt{kstavrop@cs.utexas.edu}. Supported by the NSF AI Institute for Foundations of Machine Learning (IFML) and by scholarships from Bodossaki Foundation and Leventis Foundation.} \\ UT Austin
    \and
    Arsen Vasilyan\thanks{\texttt{arsenvasilyan@gmail.com}.  Supported in part by NSF awards CCF-2006664, DMS-2022448, CCF-1565235, CCF-1955217, CCF-2310818, Big George Fellowship and Fintech@CSAIL. Part of this work was conducted while the author was visiting the Simons Institute for the Theory of Computing. Work done in part while visiting UT Austin.} \\ UC Berkeley
}
\date{}

\begin{document}

\maketitle

\begin{abstract}
    The seminal work of Linial, Mansour, and Nisan gave a quasipolynomial-time algorithm for learning constant-depth circuits ($\mathsf{AC}^0$) with respect to the uniform distribution on the hypercube.  Extending their algorithm to the setting of malicious noise, where both covariates and labels can be adversarially corrupted, has remained open.  Here we achieve such a result, inspired by recent work on learning with distribution shift. Our running time essentially matches their algorithm, which is known to be optimal assuming various cryptographic primitives.
    
    Our proof uses a simple outlier-removal method combined with Braverman's theorem for fooling constant-depth circuits.  We attain the best possible dependence on the noise rate and succeed in the harshest possible noise model (i.e., contamination or so-called ``nasty noise"). 
\end{abstract}

\ignore{
\begin{abstract}
    The problem of learning under malicious noise was defined more than two decades ago by Bshouty, Eiron and Kushilevitz \cite{BSHOUTY2002255}, and yet there are no known non-trivial algorithms for learning even the class of DNFs under the uniform distribution in this model. We propose a very simple algorithmic approach that yields quasi-polynomial time complexity upper bounds for circuits of any constant depth. In particular, we achieve runtime of $d^{O(k)}$, where $k = {(\log(s))^{O(\ell)}\log(1/\eps)}$, for learning depth-$\ell$ circuits of size $s$ over $\cube{d}$ up to excess error $\eps$, matching the best known results for the easier task of agnostic learning. We combine a refined version of standard outlier removal procedures with the low-degree sandwiching approximators of Braverman \cite{braverman2008polylogarithmic}.
\end{abstract}
}
\thispagestyle{empty}
\newpage
\setcounter{page}{1}
\section{Introduction}
In their famous paper, Linial, Mansour, and Nisan \cite{linial1993constant} introduced the ``low-degree" algorithm for learning Boolean functions with respect to the uniform distribution on $\cube{d}$.  The running time and sample complexity of their algorithm scales in terms of the Fourier concentration of the underlying concept class, and, using this framework, they obtained a quasipolynomial-time algorithm for learning constant-depth, polynomial-size circuits ($\mathsf{AC}^0$).

Prior work \cite{kalai2008agnostically} had extended their result to the agnostic setting, where the {\em labels} can be adversarially corrupted, but the marginal distribution on inputs must still be uniform over $\cube{d}$. Remarkably, there had been no progress on this problem in the last three decades for {\em malicious} noise models where {\em both} covariates and labels can be adversarially corrupted \cite{Valiant85,KearnsL93}.  

In this paper, we completely resolve this problem and obtain a quasipolynomial-time algorithm for learning $\mathsf{AC}^0$ in the harshest possible noise model, the so-called ``nasty noise" model of \cite{BSHOUTY2002255}.  We define this model below and refer to it simply as learning with contamination, in line with recent work in computationally efficient robust statistics (see e.g., \cite{DKBook}).

\begin{definition}[Learning from Contaminated Samples]\label{definition:learning-from-contaminated-samples}
    A set of $N$ labeled examples $\Sinplabeled$ is an $\eta$-contaminated (uniform) sample with respect to some class $\C \subseteq\{\cube{d} \to \cube{}\}$, where $N\in \nats$ and $\eta\in(0,1)$, if it is formed by an adversary as follows.
    \begin{enumerate}
        \item The adversary receives a set of $N$ clean i.i.d. labeled examples $\Sclnlabeled$, drawn from the uniform distribution over $\cube{d}$ and labeled by some unknown concept $\copt$ in $\C$. \item The adversary removes an arbitrary set $\Sremlabeled$ of $\lfloor\eta N\rfloor$ labeled examples from $\Sclnlabeled$ and substitutes it with an adversarial set of $\lfloor\eta N\rfloor$ labeled examples $\Sadvlabeled$.
    \end{enumerate} 
    Namely, $\Sinplabeled = (\Sclnlabeled\setminus \Sremlabeled)\cup \Sadvlabeled$. For the corresponding unlabeled set $\Sinp$, we say that it is an $\eta$-contaminated (uniform) sample.

    In this model, 
    the goal of the learner is to output (with probability $1-\delta$) a hypothesis $h:\cube{d}\to \cube{}$ such that $\pr_{\x\sim\Unif(\cube{d})}[h(\x) \neq \copt(\x)] \le 2\eta + \epsilon$. {The factor $2$ is known to be the best possible constant achievable by any algorithm \cite{BSHOUTY2002255}. 
    }
\end{definition}

Although there is now a long line of research giving computationally efficient algorithms for learning Boolean function classes in malicious noise models, these algorithms primarily apply to geometric concept classes and continuous marginal distributions, such as halfspaces or intersections of halfspaces with respect to Gaussian or log-concave densities \cite{kalai2008agnostically,klivans2009learning,awasthi2017power,diakonikolas2018learning,zhangmal}.  In particular, nothing was known for the case of $\mathsf{AC^0}$.


Our main theorem is as follows:

\begin{theorem}\label{theorem:main-ac0}
    For any $s,\ell, d \in \N$, and $\eps,\delta\in(0,1)$, there is an algorithm that learns the class of $\mathsf{AC}^0$ circuits of size $s$ and depth $\ell$ and achieves error $2\eta + \epsilon$, with running time and sample complexity $d^{O(k)}\log(1/\delta)$, where $k = {(\log(s))^{O(\ell)}\log(1/\eps)}$, from contaminated samples of any noise rate $\eta$.
\end{theorem}

Our running time essentially matches the Linial, Mansour, and Nisan result, which is known to be optimal assuming various cryptographic primitives \cite{Kharitonov95}. 

More generally, we prove that any concept class $\C$ that admits $\ell_1$-{\em sandwiching polynomials} of degree $k$ can be learned in time $d^{O(k)}$ from contaminated samples.  Recent work due to \cite{goel2024tolerant} had obtained a similar result achieving the weaker bound of $O(\eta) + \epsilon$ for learning functions with $\ell_2$-sandwiching polynomials.
Crucially, it remains unclear how to obtain such $\ell_2$ sandwiching approximators for constant depth circuits \footnote{Braverman's celebrated result on $\mathsf{AC^0}$ \cite{braverman2008polylogarithmic} obtains only $\ell_1$-sandwiching.}, and so their result does not apply here. 

In 2005, Kalai et al. \cite{kalai2008agnostically} showed that $\ell_1$-approximation suffices for agnostic learning.  Here we complete the analogy for malicious learning, showing that $\ell_1$-sandwiching implies learnability with respect to contamination.


\paragraph{Proof Overview.}
The input set $\Sinplabeled$ is $\eta$-contaminated. This might make it hard to find a hypothesis with near-optimal error on $\Sinplabeled$. However, we are only interested in finding a hypothesis with error $2\eta+\eps$ on the clean distribution, which is structured (in particular, the marginal distribution on the features is uniform over $\cube{d}$). In order to take advantage of the structure of the clean distribution despite only having access to the contaminated sample, we make use of the notion of sandwiching polynomials:
\begin{definition}[Sandwiching polynomials]\label{definition:sandwiching-polynomials}
    Let $f:\cube{d} \to \cube{}$. We say that the ($\ell_1$) $\eps$-sandwiching degree of $f$ with respect to the uniform distribution over the hypercube $\cube{d}$ is $k$ if there are polynomials $\pup,\pdown:\cube{d}\to \R$ of degree at most $k$ such that (1) $\pdown(\x) \le f(\x)\le \pup(\x)$ for all $\x\in\cube{d}$ and (2) $\E_{\x\sim \unif(\cube{d})}[\pup(\x) - \pdown(\x)] \le \eps$.
\end{definition}
The sandwiching degree of size-$s$ depth-$\ell$ $\mathsf{AC}^0$ circuits is bounded by $k = (\log(s))^{O(\ell)}\log(1/\eps)$, due to the result of Braverman on fooling constant-depth circuits (see \Cref{theorem:sandwiching-ac0} from \cite{braverman2008polylogarithmic,tal:LIPIcs.CCC.2017.15,harsha2019polynomial}).
Suppose that $\bar S$ is a subset of $\Sinplabeled$ that preserves the expectations of low-degree and non-negative polynomials (e.g., $\pup-\pdown$) compared to the uniform distribution.
Under this condition, low-degree polynomial regression gives a hypothesis with near-optimal error on $\bar S$ (see \Cref{section:low-error}). 


We show in \Cref{lemma:outlier-removal} that a simple procedure that iteratively removes samples from $\Sinplabeled$ can be used to form such a set $\bar S$ (that preserves the expectations of non-negative, degree-$k$ and low-expectation polynomials) and, moreover, this procedure removes more contaminated points than clean points. The last property is important, because it implies that $\bar S$ is representative for the ground truth distribution, i.e., any near-optimal hypothesis for $\bar S$ will also have error $2\eta + \eps$ on the ground truth. 

This is possible because the only way the adversary can significantly increase the expectation of a non-negative polynomial $p$ is by inserting examples $\x$ where $p(\x)$ is unreasonably large compared to the typical values of $p$ over the uniform distribution. Our algorithm iteratively finds the non-negative polynomial $q$ with the largest expectation over a given set through a simple linear program and then removes the points $\x$ for which $q(x)$ is large.

Our iterative outlier removal procedure is inspired by prior work on TDS learning (Testable Learning with Distribution Shift) and PQ learning \cite{klivans2023testable,goel2024tolerant} as well as the work of \cite{diakonikolas2018learning} on learning geometric concepts from contaminated examples. Both of these works use outlier removal procedures that give bounds on the variance of polynomials rather than the expectation of non-negative polynomials and, instead of linear programming, they use spectral algorithms.

\section{Notation}

Throughout this work, when we refer to a set $S$ of examples from the hypercube $\cube{d}$, we consider every example in $S$ to be a unique and separate instance of the corresponding element in $\cube{d}$. Moreover, we denote with $\bar{S}$ the corresponding labeled set of examples in $\cube{d}\times\cube{}$. 

{Recall that }
polynomials over $\cube{d}$ 
{ are} functions of the form $p(\x) = \sum_{\I\subseteq[d]}c_p(\I) \prod_{i\in \I}x_i$, where $\x = (x_i)_{i\in[d]}$ and $c_p(\I)\in \R$. We denote with $\x^\I$ the quantity $\prod_{i\in \I}x_i$. We say that the degree of $p$ is at most $k$ if for any $\I\subseteq[d]$ with $|\I| > k$, we have $c_p(\I) = 0$. For a polynomial $p$, we denote with $\|p\|_{\coef}$ the $\ell_1$ norm of its coefficients, i.e., $\|p\|_\coef = \sum_{\I\subseteq[d]}|c_p(\I)|$.

\section{Removing the Outliers}\label{section:outlier-removal}

The input set $\Sinplabeled$ includes an $\eta$ fraction of contaminated examples. It is, of course, impossible to identify the exact subset of $\Sinplabeled$ that is contaminated. 
However, we show how to remove contaminated examples that lead to inflation of the expectations of low-degree non-negative polynomials, which we call ``outliers." We remove only a relatively small number of clean examples from $\Sinplabeled$, as we show in the following lemma.

\begin{lemma}[Outlier removal]\label{lemma:outlier-removal}
    Let $\Sinp$ be an $\eta$-contaminated uniform sample (see \Cref{definition:learning-from-contaminated-samples}) with size $N$. For any choice of the parameters $\epsilon,\delta\in (0,1)$, and $k\in \nats$, the output $\Sout$ of \Cref{algorithm:main} satisfies the following, whenever $N \ge C\frac{(3d)^{2k}}{\eps^2} \log(1/\delta)$, for some sufficiently large constant $C\ge 1$.
    \begin{enumerate}
        \item\label{case:monotonicity} With probability at least $1-\delta$, the number of clean examples in $\Sinp$ that are removed from $\Sout$ is at most equal to the number of adversarial examples that are removed from $\Sout$ (see \Cref{fig:contaminated-diagram}). Namely, $|(\Sinp\cap\Scln)\setminus\Sout| \le |\Sadv\setminus\Sout|$.
        \item\label{case:poly-bound} For any non-negative polynomial $p$ over $\cube{d}$ with degree at most $k$ 
        and $\E_{\x\sim\Unif_d}[p(\x)]\le \frac{\eps}{8}$, we have
        $
            \sum_{\x\in\Sout} p(\x) \le \epsilon N \text{ with probability at least }1-\delta.
        $
    \end{enumerate}
\end{lemma}

\begin{algorithm}
\caption{Outlier removal through Linear Programming}\label{algorithm:main}
\KwIn{Set $\Sinp\subseteq \cube{d}$ of size $N$ and parameters $\eps\in (0,1)$, $B>0$ and $k \in \nats$}
\KwOut{Filtered set $\Sout \subseteq \Sinp$.}
Let $B = 3^k d^{k/2}$, $\Delta = \frac{\eps}{2B}$.
\\
$S^{(0)} \leftarrow \Sinp$\\
\For{$i = 0, 1, 2, \dots, N$}{
    Let $p^*$ be the solution of the following linear program (P) and $\lambda^* = \frac{1}{N}\sum_{\x\in S^{(i)}}p^*(\x)$. 
    \begin{equation}
    \left.
    \begin{cases}
        \max_{p}& \,\sum_{\x\in S^{(i)}}p(\x) \hspace{1em} \notag \\
        \text{ s.t.:}&\,\, p \text{ polynomial, }\deg(p)\le k \text{ and }\|p\|_{\coef} \le B \notag\\
        &\,\, p(\x) \ge 0, \text{ for all }\x\in\Sref\cup \Sinp \notag \\
        &\,\, \frac{1}{N}\sum_{\x\in \Sref}p(\x) \le \eps/4 \notag
    \end{cases}
    \right\} \tag{P}
    \end{equation}\\
\lIf{$\lambda^* \le \eps$}{output $\Sout \leftarrow S^{(i)}$ and terminate}
\Else{let $\tau^*\ge 0$ be the smallest value such that $\frac{|S^{(i)}|}{N}\pr_{\x\sim S^{(i)}}[p^*(\x) > \tau^*] \ge 2\pr_{\x\sim \Sref}[p^*(\x) > \tau^*]+\Delta$
    \\
    $S^{(i+1)} \leftarrow S^{(i)} \setminus \{\x\in S^{(i)}: p^*(\x) > \tau^*\}$
}
}
\end{algorithm}

Our \Cref{algorithm:main} is similar in spirit to outlier removal procedures that have been used previously in the context of learning with contaminated samples \cite{diakonikolas2018learning} and tolerant learning with distribution shift \cite{goel2024tolerant}: we iteratively find the non-negative polynomial with largest expectation and remove the examples that give this polynomial unusually large values. Here we focus on the expectations of non-negative polynomials, while in all previous works, the guarantees after outlier removal concerned the variance of arbitrary polynomials. In this sense, our guarantees are stronger, but only hold for non-negative polynomials. Our algorithm solves, in every iteration, one linear program (P) in place of the usual spectral techniques from prior work. 

The proof idea is that whenever there is a non-negative polynomial $p^*$ with unreasonably large expectation, there have to be many outliers that give unusually large values to $p^*$. By removing all the points where $p^*$ is large, we can, therefore, be confident that we remove more outliers than clean examples (part 1 of \Cref{lemma:outlier-removal}). When the algorithm terminates, all non-negative polynomials with low expectation under the uniform distribution, will also have low expectation under the remaining set of examples (part 2 of \Cref{lemma:outlier-removal}).

    For part 1, we analyze the non-terminating iterations and we show that for each clean point that is filtered out by the procedure, at least one adversarial point is filtered out as well. We first show that in such an iteration, there is a $\tau\in[0,B]$ such that $\frac{|S^{(i)}|}{N}\pr_{\x\sim S^{(i)}}[p^*(\x) > \tau] \ge 2\pr_{\x\sim \Sref}[p^*(\x) > \tau]+\Delta > 0$. This implies that in every non-terminating iteration at least one point is removed and, therefore, some iteration $i\le N$ will satisfy the stopping criterion and terminate (there are only $N$ points in total).

    \begin{claim}
        In any non-terminating iteration ({i.e. an iteration where} $\lambda^* > \eps$), there is $\tau^*\in [0,B]$ such that \[\frac{|S^{(i)}|}{N}\pr_{\x\sim S^{(i)}}[p^*(\x) > \tau^*] \ge 2\pr_{\x\sim \Sref}[p^*(\x) > \tau^*]+\Delta.\]
    \end{claim}

    \begin{proof}
    Suppose, for contradiction, that for all $\tau\in[0,B]$ we have 
    \[
        \frac{|S^{(i)}|}{N}\pr_{\x\sim S^{(i)}}[p^*(\x) > \tau] < 2\pr_{\x\sim \Sref}[p^*(\x) > \tau]+\Delta
    \]
    We may integrate over $\tau\in[0,B]$ both sides of the above inequality, since the corresponding functions of $\tau$ have finite number of discontinuities (at most equal to $|S^{(i)}|+|\Sref|$).
    \begin{equation}\label{equation:int-bound}
        \frac{|S^{(i)}|}{N} \int_{\tau = 0}^B\pr_{\x\sim S^{(i)}}[p^*(\x) > \tau]\, d\tau < 2\int_{\tau = 0}^B\pr_{\x\sim \Sref}[p^*(\x) > \tau]\, d\tau + \Delta B
    \end{equation}
    We will now substitute the integrals above with expectations, i.e., $\int_{\tau = 0}^B\pr_{\x\sim S^{(i)}}[p^*(\x) > \tau]\, d\tau = \E_{\x\sim S^{(i)}}[p^*(\x)]$ and $\int_{\tau = 0}^B\pr_{\x\sim \Sref}[p^*(\x) > \tau]\, d\tau = \E_{\x\sim \Sref}[p^*(\x)]$. We use the simple fact that for any non-negative random variable $X$ with values in $[0,B]$, we have $\E[X] = \int_{\tau = 0}^B \pr[X>\tau]\, d\tau$. 

    We first set $X = p^*(\x)$, where $\x\sim S^{(i)}$ and observe that (1) $p^*(\x) \ge 0$ for all $\x\in S^{(i)}$, and also that (2) $p^*(\x) \le \|p\|_\coef \le B$ for all $\x\in\cube{d}\supseteq S^{(i)}$, since $p^*$ satisfies $\|p\|_\coef \le B$ according to the constraints of (P) and $p^*(\x) = \sum_{\I\subseteq[d]}c_{p^*}(\I) \x^\I \le \sum_{\I\subseteq[d]}|c_{p^*}(\I)| \cdot|\x^\I| = \sum_{\I\subseteq[d]}|c_{p^*}(\I)| = \|p^*\|_\coef$, since $\x\in\cube{d}$ and therefore $|\x^\I| = 1$. This shows that $X \in [0,B]$ almost surely over $\x\sim S^{(i)}$. Using an analogous argument for $\x\sim \Sref$, we overall obtain the following.
    \begin{equation}\label{equation:int-to-exp}
        \int_{\tau = 0}^B\pr_{\x\sim S^{(i)}}[p^*(\x) > \tau]\, d\tau = \E_{\x\sim S^{(i)}}[p^*(\x)] \;\;\text{ and }\;\;\int_{\tau = 0}^B\pr_{\x\sim \Sref}[p^*(\x) > \tau]\, d\tau = \E_{\x\sim \Sref}[p^*(\x)]
    \end{equation}
    We may now substitute \eqref{equation:int-to-exp} in the inequality \eqref{equation:int-bound}, and use the fact that $\E_{\x\sim\Sref}[p^*(\x)]\le \eps/4$ (by the constraints of (P)) to conclude that
    \[
        \lambda^* = \frac{1}{N} \sum_{\x\sim S^{(i)}}p^*(\x) = \frac{|S^{(i)}|}{N} \E_{\x\sim S^{(i)}}[p^*(\x)] \le 2\eps/4 + \eps/2 = \eps
    \]
    We reached a contradiction, since $\lambda^*>\eps$, and, therefore, $\tau^*$ exists.
    \end{proof}

    We still need to show that whenever the procedure filters out clean examples, it also filters out an equal number of adversarial examples. Let $S_r^{(i)} = \{\x\in S^{(i)}: p^*(\x) >\tau^*\}$ be the set of points that are filtered out during iteration $i$. We can write $S_r^{(i)}$ as a disjoint union $S_{r,\mathrm{cln}}^{(i)} \cup S_{r,\mathrm{adv}}^{(i)}$, where $S_{r,\mathrm{cln}}^{(i)} = S_r^{(i)} \cap \Scln$ are the clean examples that are removed and $S_{r,\mathrm{adv}}^{(i)} = S_r^{(i)} \cap \Sadv$ are the adversarial examples that are removed.

    \begin{claim}
        With probability at least $1-\delta$, we have that for all non-terminating iterations, $|S_{r,\mathrm{cln}}^{(i)}| \le |S_{r,\mathrm{adv}}^{(i)}|$.
    \end{claim}

    \begin{proof}
        By the previous claim, we know that $\tau^*$ (which defines the set $S_r^{(i)}$) exists and has the property that $\frac{|S^{(i)}|}{N}\pr_{\x\sim S^{(i)}}[p^*(\x) > \tau^*] \ge 2\pr_{\x\sim \Sref}[p^*(\x) > \tau^*]+\Delta$. 

        We first focus on the quantity $\pr_{\x\sim \Sref}[p^*(\x) > \tau^*]$, which is proportional to the number of reference examples that would be removed by the thresholding operation $p^*(\x)>\tau^*$. However, we are interested in the number of actual clean examples that would be removed. The reference examples can be shown to provide an estimate of the number of removed clean examples, through uniform convergence. In particular, the thresholding operation corresponds to a polynomial threshold function of degree at most $d^k$ and, therefore, by standard VC dimension arguments (and uniformly for all iterations) we have that $\pr_{\x\sim \Sref}[p^*(\x) > \tau^*] \ge \pr_{\x\sim \Scln}[p^*(\x) > \tau^*] - \Delta/2$, except with probability $\delta$ 
        , as long as the sample size is $N \ge C' \frac{d^k+\log(1/\delta)}{\Delta^2}$. This is because both $\Sref$ and $\Scln$ consist of $N$ i.i.d. samples from the uniform distribution. 
        
        Overall, we have that $\frac{|S^{(i)}|}{N}\pr_{\x\sim S^{(i)}}[p^*(\x) > \tau^*] \ge 2\pr_{\x\sim \Scln}[p^*(\x) > \tau^*]$. We can write the empirical probabilities in terms of the sizes of the removed sets to obtain the following, where we also use the fact that $|S_r^{(i)}| = |S_{r,\mathrm{cln}}^{(i)}| + |S_{r,\mathrm{adv}}^{(i)}|$ and that $|S_{r,\mathrm{cln}}^{(i)}|$ is at most equal to the number of clean examples that would be removed by the $i$-th filtering operation (some clean examples could already have been removed either by the adversary or by some previous iteration and these will not be contained in $S_{r,\mathrm{cln}}^{(i)}$).
        \begin{align*}
            &\frac{|S^{(i)}|}{N}\cdot\frac{|S_r^{(i)}|}{|S^{(i)}|} \ge 2 \frac{|S_{r,\mathrm{cln}}^{(i)}|}{N}
            \;\;\text{ or }\;\; |S_{r,\mathrm{cln}}^{(i)}| + |S_{r,\mathrm{adv}}^{(i)}| \ge 2|S_{r,\mathrm{cln}}^{(i)}| 
            \;\;\text{ or }\;\; |S_{r,\mathrm{adv}}^{(i)}| \ge |S_{r,\mathrm{cln}}^{(i)}|
        \end{align*}
        This concludes the proof of the claim.
    \end{proof}

    Overall, if we sum over $i\in[N]$, we obtain that the number of clean examples that are removed by the procedure is at most equal to the number of adversarial examples that are removed by the procedure.

    For part 2 of \Cref{lemma:outlier-removal}, we observe that $\sum_{\x\in\Sout} p(\x) \le \lambda^* N \le \eps N$, as long as $p$ satisfies all the constraints of the program (P). It suffices to prove the following claim.

    \begin{claim}
        Any non-negative polynomial $p$ with degree at most $k$ and $\E_{\x\sim \unif_d}[p(\x)] \le \eps/8$ satisfies all the constraints of the program (P) with probability at least $1-\delta$.
    \end{claim}

    \begin{proof}
        The degree bound and non-negativity are satisfied directly by the definition of $p$. We now need to show that $\|p\|_\coef \le 3^k d^{k/2}$. Recall that $p(\x) = \sum_{\I\subseteq [d]}c_p(\I) \x^S$, where $c_p(\I) = 0$ for any $|\I| > k$ and $\|p\|_\coef = \sum_{\I\subseteq [d]} |c_p(\I)|$. By viewing $c_p$ as a vector with $\sum_{j=0}^k\binom{d}{j}\le d^k$ dimensions (assuming $2\le k\le d$), we have that $\|p\|_\coef = \|c_p\|_1 \le d^{k/2} \|c_p\|_2 = d^{k/2} (\E_{\x\sim \unif_d}[(p(\x))^2])^{1/2}$.

        Since the uniform distribution is $(2,1)$-hypercontractive (see Theorem 9.22 in \cite{o2014analysis}), we have
        \begin{equation}\label{equation:hypercontractivity}
            \E_{\x\sim \unif_d}[(p(\x))^2])^{1/2} \le e^k \E_{\x\sim \unif_d}[|p(\x)|]
        \end{equation}
        Recall that the polynomial $p$ is non-negative. This implies that $|p(\x)| = p(\x)$ for all $\x\in\cube{d}$ and therefore $\E_{\x\sim \unif_d}[|p(\x)|] = \E_{\x\sim \unif_d}[p(\x)] \le \eps/8$. Overall, we have
        \begin{equation}\label{equation:hypercontractivity-applied}
            \E_{\x\sim \unif_d}[(p(\x))^2])^{1/2} \le e^k \eps/8 \le 3^k
        \end{equation}
        Recall, now, that $\|p\|_\coef \le d^{k/2} (\E_{\x\sim \unif_d}[(p(\x))^2])^{1/2}$. We obtain the desired bound $\|p\|_\coef \le 3^k d^{k/2}$.

        It remains to show that with probability at least $1-\delta$, we have $\frac{1}{N}\sum_{\x\in \Sref}p(\x) \le \eps/4$. Consider the random variable $X = \frac{1}{N}\sum_{\x\in \Sref}p(\x)$, where $\Sref$ is drawn i.i.d. form $\unif_d$. We have that $\E[X] = \E_{\x\sim\unif_d}[p(\x)] \le \eps/8$. Moreover, $p(\x) \le \|p\|_\coef \le 3^k d^{k/2}$, for all $\x\in\cube{d}$ and, from a standard Hoeffding bound on the random variable $X$, we obtain that $\frac{1}{N}\sum_{\x\in \Sref}p(\x) \le \eps/4$ with probability at least $1-\exp(-\frac{\eps^2}{64N d^k 9^k})$. Due to the choice of $N \ge C'\frac{d^k 9^k}{\eps^2}\log(1/\delta)$, we have that the probability of failure is bounded by $\delta$, as desired.
    \end{proof}
    

\section{Finding a Low-Error Hypothesis}\label{section:low-error}

The outlier removal process of \Cref{lemma:outlier-removal} enables us to find a subset $\Sout$ of the input set such that all non-negative and low-degree polynomials with small expectation under the uniform distribution also have small empirical expectation under $\Sout$. Moreover, the number of clean examples removed to form $\Sout$ is smaller than the number of removed outliers (see \Cref{fig:contaminated-diagram}). We show that these two properties are all we need in order to learn constant-depth circuits with contamination (\Cref{definition:learning-from-contaminated-samples}).

In order to take advantage of \Cref{lemma:outlier-removal}, we will use two main tools. The first one is the following theorem originally proposed by \cite{kalai2008agnostically} to show that $\L_1$ polynomial regression implies agnostic learning for classes that can be approximated by low-degree polynomials.

\begin{theorem}[Learning through $\mathcal{L}_1$ polynomial regression \cite{kalai2008agnostically}]\label{theorem:kkms}
    Let $\D$ be any distribution over $\cube{d}\times \cube{}$ and $\C$ some class of concepts from $\cube{d}$ to $\cube{}$. If for each $f\in\C$ there is some polynomial $p$ over $\cube{d}$ of degree at most $k$ such that $\E_{\x\sim \D_\x}[|f(\x) - p(\x)|] \le \eps$, then there is an algorithm (based on degree-$k$ $\L_1$ polynomial regression) which outputs a degree-$k$ polynomial threshold function $h$ such that $\pr_{(\x,y)\sim \D}[y\neq h(\x)] \le \min_{f\in \C}\pr_{(\x,y)\sim \D}[y\neq f(\x)] + 2\eps$, in time $O(\frac{1}{\eps^2}) d^{O(k)} \log(1/\delta)$.
\end{theorem}

Our overall learning algorithm will first filter the input set of examples $\Sinplabeled$ using \Cref{algorithm:main} and then run the algorithm of \Cref{theorem:kkms} on the uniform distribution over the filtered set $\Soutlabeled$. All we need to show is that there is a low-degree polynomial $p$ with $\E_{\x\sim \Soutlabeled}[|f(\x) - p(\x)|] \le \eps$. This is ensured by combining part 2 of \Cref{lemma:outlier-removal} with the sandwiching approximators for constant-depth circuits originally proposed by \cite{braverman2008polylogarithmic} in the context of pseudorandomness.

\begin{theorem}[Sandwiching polynomials for $\mathsf{AC}^0$ \cite{braverman2008polylogarithmic,tal:LIPIcs.CCC.2017.15,harsha2019polynomial}]\label{theorem:sandwiching-ac0}
    Let $f:\cube{d} \to \cube{}$ be any $\mathsf{AC}^0$ circuit of size $s$ and depth $\ell$ and let $\eps\in(0,1)$. Then, there are polynomials $\pup,\pdown$ over $\cube{d}$, each of degree at most $k = (\log(s))^{O(\ell)} \cdot \log(1/\eps)$ such that (1) $\pup(\x) \ge f(\x) \ge \pdown(\x)$ for all $\x\in \cube{d}$ and (2) $\E_{\x\sim \unif_d}[\pup(\x) - \pdown(\x)] \le \eps$.
\end{theorem}

\begin{proof}[Proof of \Cref{theorem:main-ac0}]
Consider the polynomial $p = \pup-\pdown$, where $\pup,\pdown$ are the $(\eps/8)$-sandwiching polynomials of some circuit $f$ of size $s$ and depth $\ell$. Observe that $p$ is non-negative and $\E_{\x\sim\unif_d}[p(\x)] \le \eps/8$. Therefore, according to part 2 of \Cref{lemma:outlier-removal}, we have $\sum_{\x\in\Sout} p(\x) \le \eps N$. Since $\pup\ge f\ge \pdown$ we also have $\E_{\x\sim \Sout}[|f(\x)-\pdown(\x)|] \le \eps N/|\Soutlabeled|$. By \Cref{theorem:kkms}, we find $h:\cube{d}\to \cube{}$ with $\pr_{(\x,y)\sim \Soutlabeled}[y\neq h(\x)] \le \min_{f\in \C}\pr_{(\x,y)\sim \Soutlabeled}[y\neq f(\x)] + 2\eps N/|\Soutlabeled|$ or equivalently
\begin{align}\label{equation:opt-error}
    \sum_{(\x,y)\in \Soutlabeled}\ind\{y\neq h(\x)\} \le \min_{f\in \C}\sum_{(\x,y)\in \Soutlabeled}\ind\{y\neq f(\x)\} + 2\eps N
\end{align}

The error of the hypothesis $h$ on the set $\Sclnlabeled$ gives a bound on its error under the uniform distribution with high probability, due to classical VC theory, as long as $N \ge C'\frac{d^k+\log(1/\delta)}{\eps^2}$, because $h$ is a PTF of degree at most $k$. We can provide an upper bound for $\pr_{(\x,y)\sim \Sclnlabeled}[y\neq h(\x)]$ in terms of the sizes of the sets depicted in \Cref{fig:contaminated-diagram}. In particular, we give a high-probability upper bound on the number of mistakes that $h$ makes on $\Sclnlabeled$. 
\begin{enumerate}
    \item The points in $\Sclnlabeled$ that are removed by the adversary are not taken into account while forming $h$, so, in the worst case, $h$ classifies them incorrectly. This gives at most $|\Scln\setminus \Sinp|$ mistakes.
    \item Similarly, $h$ makes at most $|S_3|$ mistakes corresponding to the clean points that are removed during the outlier removal process.
    \item Finally, $h$ will make at most $|S_2|+2\eps N$ mistakes on $\Sout$, according to the inequality \eqref{equation:opt-error}, corresponding to the adversarially corrupted points that were not removed during the outlier removal process. In the worst case, all of these mistakes are made in the part of $\Sout$ that intersects $\Scln$.
\end{enumerate}


\begin{figure}[h]
    \centering
    \includegraphics[width=0.3\linewidth]{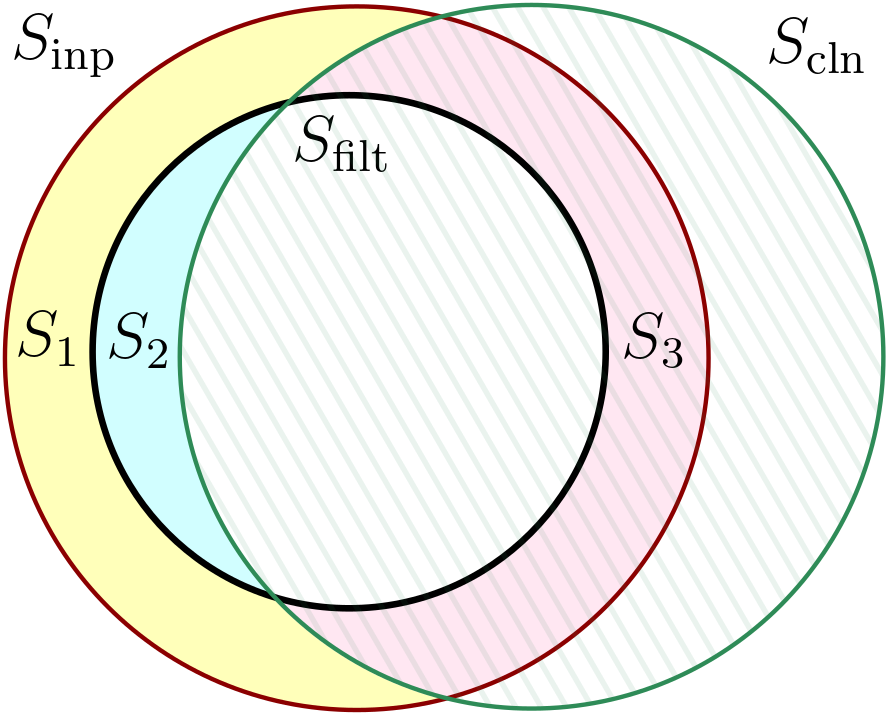}
    \caption{The diagram shows the input set of points $\Sinp$ (red circle), the clean points $\Scln$ (green circle), the output $\Sout$ (black circle) of \Cref{algorithm:main} and the sets $S_1$ (yellow region), $S_2$ (blue region), $S_3$ (pink region). The set $\Sinp$ consists of clean points, except from an $\eta$ fraction of adversarial points. $S_1$ contains the adversarial points that are filtered out by the outlier removal process and $S_2$ contains the adversarial points that were not removed and are kept in $\Sout$. $S_3$ contains the clean points that were filtered out during outlier removal. \Cref{lemma:outlier-removal} states that $|S_3| \le |S_1|$ w.h.p.}
    \label{fig:contaminated-diagram}
\end{figure}

The overall error is $\frac{1}{N} (|\Scln\setminus \Sinp| + |S_3| + |S_2|) + O(\eps)$. According to part 1 of \Cref{lemma:outlier-removal}, we have $|S_3| \le |S_1|$. Moreover, by \Cref{definition:learning-from-contaminated-samples}, we have that $|\Scln\setminus \Sinp| = |\Sinp\setminus\Scln| = |S_1|+|S_2| = \eta N$. The error bound we obtain overall is $2\eta + O(\eps)$, as desired.
\end{proof}


\paragraph{Acknowledgments.} We thank Mark Braverman and Sasha Razborov for useful conversations.

\bibliographystyle{alpha}
\bibliography{refs}



\end{document}